\documentclass[aps,prl,twocolumn,showpacs]{revtex4}
\usepackage{amsthm}
\usepackage{amsmath}
\usepackage{latexsym}
\usepackage{amsfonts}
\usepackage{amssymb}
\usepackage{color}
\usepackage{bbm,dsfont}
\usepackage{graphicx}

\usepackage{subfigure}
\usepackage{mathrsfs}
\usepackage{mathbbol}
\usepackage{hyperref}
\usepackage{enumerate}

\newtheorem{proposition?}{Proposition?}
\newtheorem{theorem}{Theorem}

\theoremstyle{definition}

\newtheorem{definition}{Definition}

\newcommand{\complex}{\mathbb C} 
\newcommand{\integer}{\mathbb Z} 
\newcommand{\half}{\tfrac{1}{2}} 

\newcommand{\hi}{\mathcal{H}} 
\newcommand{\hik}{\mathcal{K}} 
\newcommand{\kb}[2]{|#1\rangle\langle#2|} 
\newcommand{\tr}[1]{\mathrm{tr}\left[#1\right]} 
\newcommand{\ptr}[2]{\mathrm{tr}_{#1}[#2]} 
\newcommand{\id}{\mathbbm{1}} 

\newcommand{\A}{\mathsf{A}}
\newcommand{\B}{\mathsf{B}}
\newcommand{\C}{\mathsf{C}}
\newcommand{\M}{\mathsf{M}}

\newcommand{\luders}[1]{\mathcal{L}_{#1}}

\newcommand{\hin}{\hi_{in}} 
\newcommand{\hout}{\hi_{out}} 
\newcommand{\hei}[1]{{#1}^\ast} 

\begin{document}

\title[]{Universality of Sequential Quantum Measurements}

\author{Teiko Heinosaari$^\natural$}
\address{$\natural$ Turku Centre for Quantum Physics, Department of Physics and Astronomy, University of Turku, Finland}
\email{teiko.heinosaari@utu.fi}

\author{Takayuki Miyadera$^\flat$}
\address{$\flat$ Department of Nuclear Engineering, Kyoto University, 6158540 Kyoto, Japan}
\email{miyadera@nucleng.kyoto-u.ac.jp}

\begin{abstract}
Unavoidable disturbance caused by a quantum measurement implies that the realizable subsequent measurements are getting limited after one performs some measurement. 
The obvious general limitation that one cannot circumvent by sequential or any other method is that the actually implemented measurements must be jointly measurable.
In this work we show that any jointly measurable pair of observables can be obtained in a sequential measurement scheme, even if the second observable would be decided after the first measurement.
This universality feature holds only for measurement schemes with a specific structure.
As a supplementing result, we provide a characterization of all possible joint measurements obtained from a sequential measurement lacking universality.
\end{abstract}


\maketitle


Sequential quantum measurements have received renewed attention recently, and their range of application has become broader.
This concept has been investigated, for example, in the context of state discrimination \cite{BeFeHi13}, tomography \cite{CaHeTo12}, cryptography \cite{WeFrWe09,Miyadera11ijqi} and undecidability \cite{EiMuGo12}. 
However, the fundamental limitations of the sequential method of performing quantum measurements have not been addressed systematically.

In the present work we show that any jointly measurable pair of observables can be obtained via a sequential measurement scheme, even if the second observable is decided \emph{after} the first measurement.
Thus, it is possible to perform a measurement of any quantum observable in a way which does not limit the future measurements any more than is dictated by joint measurability.
This striking feature of the first measurement, which we call \emph{universality}, holds only for certain measurement schemes.
A sequential measurement scheme lacking universality allows the implementation some joint measurements, but not all that are possible.
As a supplementary result, we derive a characterization of all possible joint measurements obtained from any nonuniversal measurement scheme.

\paragraph{\bf{Sequential and joint measurements.}}

The mathematical framework for sequential quantum measurements was provided a long time ago \cite{DaLe70}.
For our purposes it suffices to use a rudimentary formulation. 
We will describe a quantum measurement as a pair $(\A,\Lambda)$ consisting of an observable $\A$ and a channel $\Lambda$, where $\A$ assigns a probability distribution of measurement outcomes to each input state and $\Lambda$ maps each input state $\varrho$ into an output state $\Lambda(\varrho)$; see Fig.  \ref{fig:measurement}.
One can obviously give more detailed descriptions of a quantum measurement, but this kind of simple description is enough for our present purposes \footnote{Two common levels of descriptions are those given by instruments and measurement models. An overview can be found in \cite{MLQT12}.}.

Mathematically, an observable (POVM) is presented as a function $\A:x\mapsto \A(x)$ from a finite set of measurement outcomes $\Omega_\A \subset \integer$ to the set of positive of operators on an input Hilbert space $\hin$, and this function must satisfy the normalization constrain $\sum_{x\in\Omega_\A} \A(x)=\id$, where $\id$ is the identity operator on $\hin$ \footnote{We will concentrate on observables with a finite number of outcomes. Generally, an observable is presented as a positive operator valued measure; see e.g. \cite{MLQT12}}.
The probability of obtaining an outcome $x$ for an input state $\varrho$ is $\tr{\varrho\A(x)}$.

A channel $\Lambda$ is presented as a completely positive and trace preserving linear map on Hilbert space operators. 
It transforms an input state $\varrho$ on $\hin$ into an output state $\Lambda(\varrho)$ on $\hout$.
We allow $\hout$ to be different than $\hin$; physically this amounts to either including an environment in the description of the output system, or to discarding some part of the input system. 

An important point is that not every pair of an observable and channel gives a valid description of a quantum measurement.
Specifically, a channel $\Lambda$ and an observable $\A$ can describe the same quantum measurement if and only if there exist completely positive maps $\Phi_x$ such that
\begin{equation}\label{eq:Kraus}
\sum_x \Phi_x = \Lambda \quad \textrm{and} \quad \tr{\Phi_x(\varrho)}=\tr{\varrho \A(x)}
\end{equation}
for all outcomes $x$ and input states $\varrho$ \cite{Ozawa84}; in this case we say that $\Lambda$ is \emph{$\A$-channel}.
Perhaps the most commonly used $\A$-channel is the L\"uders channel $\luders{\A}$ of $\A$, which is defined as $\luders{\A}(\varrho)=\sum_x \sqrt{\A(x)} \varrho \sqrt{\A(x)}$.

By a \emph{sequential measurement} we mean a setting where two or more measurements are performed on the same system, one after the other.
We will concentrate on the case of two measurements.
The first measurement must be described as an observable-channel pair $(\A,\Lambda)$, while for the second measurement it is enough to specify just the observable $\B$ since we do not examine the output state after the second measurement. 
If the initial state of the system is $\varrho$, then the obtained measurement outcome distributions are $\tr{\varrho \A(x)}$ and $\tr{\Lambda(\varrho) \B(y)}$.
We are typically interested in the properties of the input state $\varrho$ rather than the output state $\Lambda(\varrho)$, hence it is convenient to use the Heisenberg picture and write the second measurement outcome distribution as $\tr{\varrho \hei{\Lambda}(\B(y))}$, where $\hei{\Lambda}$ is the adjoint action of $\Lambda$ on the set of observables \footnote{If $\Lambda$ is a channel written in the Schr\"odinger picture, then $\hei{\Lambda}$ is defined by the formula $\tr{\Lambda(\varrho)T}=\tr{\varrho \hei{\Lambda}(T)}$, required to hold for all states $\varrho$ and operators $T$.}.
In essence, a sequential measurement of $\A$ and $\B$ gives measurement outcomes of $\A$ and $\hei{\Lambda}(\B)$  on the input state $\varrho$.

A concept related to sequential measurements is that of a joint measurement.  
Two observables $\A$ and $\B$ are \emph{jointly measurable} if there exists an observable $\M$ having the product set $\Omega_{\A}\times\Omega_{\B}$ as the set of measurement outcomes and satisfying
\begin{align}\label{eq:joint}
\A(x)=\sum_y \M(x,y) \, , \quad \B(y)=\sum_x \M(x,y) 
\end{align}
for all $x\in\Omega_{\A}$ and $y\in\Omega_{\B}$.
Any observable $\M$ satisfying \eqref{eq:joint} is called a \emph{joint observable} of $\A$ and $\B$.
This definition easily extends to any finite number of observables; $\A,\B,\C,\ldots,$ are jointly measurable if there is a single observable $\M$ whose marginals coincide with  $\A,\B,\C,\ldots$.
Note that the specification of measurement dynamics is not needed in the definition of joint measurability, unlike in that of sequential measurement. 

A sequential measurement of two quantum observables can be seen as a special type of joint measurement.
Formally, if we have maps $\Phi_x$ satisfying \eqref{eq:Kraus}, then we can define $\M(x,y)=\hei{\Phi}_x(\B(y))$.
This is a joint observable of $\A$ and the perturbed version $\hei{\Lambda}(\B)$ of $\B$.
At first sight, joint measurement is a broader concept than sequential measurement; a joint measurement is any type of measurement from which one can extract the desired probability distributions of measurement outcomes, whereas in a sequential measurement one has to measure two observables, one after the other. 
Hence, an immediate question arises:
does the sequential method of measuring two quantum observables suffer from limitations specific to it, or can one perform all possible joint measurements in this way?
In the following we will show that there are no additional limitations, and there is even a surprising advantage in certain kinds sequential measurements, which we will call \emph{universality}.

\begin{figure}
    \centering
    
         \includegraphics[width=6cm]{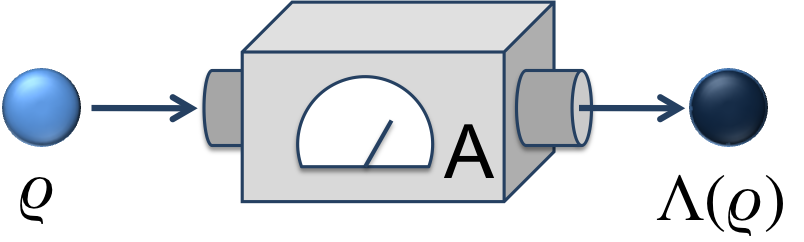}     
   
        \caption{(Color online)    A quantum measurement leads to a classical output (measurement outcome) and quantum output (transformed quantum state).
        The respective mappings $\A$ and $\Lambda$ are called \emph{observable} and \emph{channel}. }
    \label{fig:measurement}
\end{figure}

\paragraph{\bf{Sharp observables.}}

As a warm up, let us assume that $\A$ is a \emph{sharp observable}, i.e., $\A$ is a POVM and each operator $\A(x)$ is a projection. If we perform a standard von Neumann measurement of $\A$, then the state transformation is described by the L\"uders channel of $\A$. 
Supposing that the subsequently measured observable is $\B$, then the actually implemented perturbed version is given by $y \mapsto \sum_x \A(x)\B(y)\A(x)$. 
We see that if $\A$ and $\B$ commute (i.e. $\A(x)\B(y)=\B(y)\A(x)$ for all $x,y$), then $\sum_x \A(x)\B(y)\A(x)=\B(y)$ and this sequential measurement is a joint measurement of $\A$ and $\B$. 
On the other hand, it is well known that a sharp observable is jointly measurable with another observable if and only if they commute; see e.g. \cite{HeReSt08}. 
We conclude that a joint measurement of a sharp observable $\A$ and some other observable $\B$ can always be implemented as a sequential measurement of $\A$ and $\B$. 

The previously described case is a particular instance of the class of measurement schemes where the first measurement does not disturb the second one at all; the \emph{nondisturbance condition} requires that
\begin{align}\label{eq:nondisturbance}
\tr{\Lambda(\varrho)\B(y)}= \tr{\varrho \B(y)}
\end{align}
for all input states $\varrho$ and outcomes $y$.
Obviously, if the nondisturbance condition holds, then a sequential measurement described by $(\A,\Lambda)$ and $\B$ implements a joint measurement of $\A$ and $\B$ even if $\A$ is not sharp. 
The condition \eqref{eq:nondisturbance} holds, for instance, if $\A$ and $\B$ commute and we choose $\Lambda$ to be the L\"uders channel of $\A$. However, for observables that are not sharp, the condition \eqref{eq:nondisturbance} may be fulfilled for some $\A$-channel even if  $\A$ and $\B$ do not commute, and joint measurability may hold even if there is no nondisturbing measurement at all \cite{HeWo10}. 
A special feature of sharp observables is the equivalence of commutativity, nondisturbance and joint measurability.

\paragraph{\bf{Channels having the universal property.}}

\begin{figure}
    \centering
    \subfigure[]
    {
         \includegraphics[width=3.5cm]{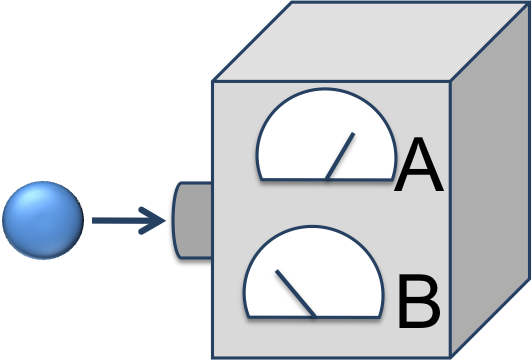}      }
    \subfigure[]
    {
        \includegraphics[width=7cm]{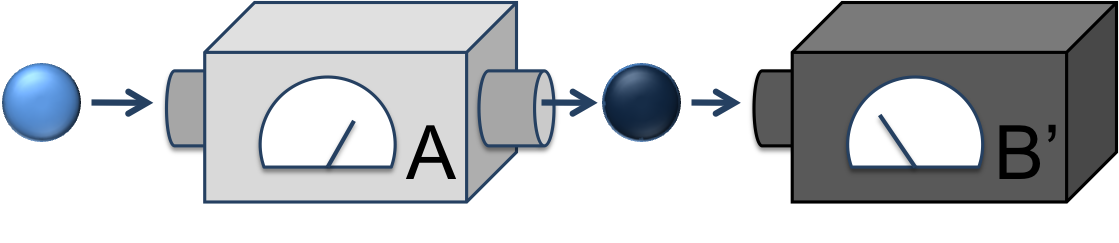} 

    }
        \caption{(Color online)    (a) A joint measurement of observables $\A$ and $\B$ gives measurement outcomes for both $\A$ and $\B$.         (b) Sequential measurement of $\A$ and a modified observable $\B'$ is equivalent to the joint measurement if the modification properly compensates for the disturbance that the first measurement causes.}
    \label{fig:joint-and-seq}
\end{figure}

In the general case, the first measurement disturbs the second one, but a joint measurement may still be possible.
In practice, this means that to obtain a joint measurement of some observables $\A$ and $\B$, we may need to measure first $\A$ and then $\B'$, which is a modified version of $\B$.
The modification is aimed to compensate the disturbance that the first measurement causes; see Fig. \ref{fig:joint-and-seq}.
We are thus looking for an $\A$-channel $\Lambda$ and an observable $\B'$ such that
\begin{align}\label{eq:sequential}
\tr{\Lambda(\varrho)\B'(y)}=\tr{\varrho \B(y)}
\end{align}
for all input states $\varrho$ and outcomes $y$, or equivalently in the Heisenberg picture, 
\begin{align}\label{eq:sequential-2}
\hei{\Lambda}(\B'(y))=\B(y)
\end{align}
for all outcomes $y$.

These equations can be interpreted in two ways. 
First, we may measure some observable $\B'$ after $\A$, typically as sharp, informative or good as possible, and then \eqref{eq:sequential-2} determines the actually implemented observable $\B$.
Second, we may want to obtain exactly $\B$ as the second observable, in which case we can try to tailor $\Lambda$ and $\B'$ in such a way that $\B$ is acquired.  

As a preliminary step towards our main result, we recall a simple construction which shows that the required objects $\Lambda$ and $\B'$ exist whenever $\A$ and $\B$ are jointly measurable \cite{HeWo10}.
Namely, suppose that $\A$ and $\B$ are jointly measurable, hence there exists $\M$ such that $\A(x)=\sum_y \M(x,y)$ and $\B(y)=\sum_x \M(x,y)$.
We define a channel $\Lambda^\B$ with an output Hilbert space $\hi_{out}=\complex^{|\Omega_{\B}|}$
as
\begin{equation}\label{eq:A->B}
\Lambda^\B(\varrho)
=\sum_y \mbox{tr}[\varrho \B(y)] |y \rangle \langle y|, 
\end{equation}
where $\{|y\rangle \}_{y \in \Omega_{\B}}$ is an orthonormal basis 
of $\hi_{out}$. 
Since
\begin{equation}
\Lambda^\B(\varrho)= \sum_{x} \left( \sum_y \tr{\varrho \M(x,y)}|y\rangle 
\langle y| \right)
\end{equation}
and
\begin{equation}
\mbox{tr}\left[ \sum_y \tr{\varrho \M(x,y)}|y\rangle 
\langle y| \right] 
=\mbox{tr}[\varrho \A(x)]
\end{equation}
we conclude that $\Lambda^\B$ is an $\A$-channel.
Moreover, 
\begin{equation}\label{eq:lambdab}
\mbox{tr}\left[ \Lambda^{\B}(\varrho) |y\rangle 
\langle y | \right] =\mbox{tr}[\varrho \B(y)]
\end{equation}
so that \eqref{eq:sequential} holds for the choices 
$\B'(y)=|y\rangle \langle y |$ and $\Lambda=\Lambda^\B$.
In conclusion, this sequential measurement scheme implements a joint measurement of $\A$ and $\B$.

The previously defined sequential measurement is not very useful since the applied $\A$-channel $\Lambda^\B$ is  designed specifically for $\B$.
This weakness becomes clear when we consider a collection of observables that are pairwisely jointly measurable without being jointly measurable as a whole \cite{HeReSt08,KuHeFr13} -- suppose  $\A$, $\B$ and $\C$ are such.
Hence, $\A$ is pairwisely jointly measurable with both $\B$ and $\C$, but the triplet $(\A,\B,\C)$ is not jointly measurable.
(An example of this sort of triplet is formed when the usual $x,y,z$-spin-component observables of a spin-$\half$ particle are mixed with uniform noise with a mixing parameter $t$ chosen from the interval $1/\sqrt{3} < t \leq 1/\sqrt{2}$.)
Suppose further that our task is to measure either the pair $(\A,\B)$ or $(\A,\C)$, but we will be told the desired pair only after we have performed the measurement of $\A$.
Since the triplet $(\A,\B,\C)$ is not jointly measurable, it is not clear how to succeed in this task.
In particular, the sequential measurement scheme related to $\Lambda^\B$ does not work since in that case the first measurement has to be chosen according to the second one, and 
one can see that for any observable $\C'$ on $\hi_{out}$,  we get $\hei{(\Lambda^{\B})}(\C'(z)) = \sum_y \langle y |\C'(z) |y\rangle \B(y)$, which is just a smearing of $\B$.

To be able to overcome this drawback, we need an $\A$-channel $\Lambda$ that satisfies the criterion \eqref{eq:sequential} for both $\B$ and $\C$ with some modified versions $\B'$ and $\C'$ on the left hand side, respectively.
In the best case we would have an $\A$-channel $\Lambda$ that satisfies the sequential measurement criterion \eqref{eq:sequential} for all observables that are jointly measurable with $\A$.
This motivates the following definition.

\begin{definition}
An $\A$-channel $\Lambda$ has the \emph{universal property} (relative to $\A$)
if for each observable $\B$ jointly measurable with $\A$, there exists an observable $\B'$ such that
\begin{align}\label{eq:universal}
\tr{\Lambda(\varrho) \B'(y)}=\tr{\varrho \B(y) }
\end{align}
for all input states $\varrho$ and outcomes $y\in\Omega_\B$.
\end{definition}

The universal property means that the measurement of an observable $\A$ limits the future measurements no more than is necessary, thus putting no additional limitations (i.e. other than joint measurability) on the  measurements that can be implemented later. 
Our main result states that these kinds of measurements exist.

\begin{theorem}\label{thm:universal}
For every observable $\A$, there exists an $\A$-channel $\Lambda_\A$ having the universal property.
\end{theorem}

Before giving the proof, we recall that each observable $\A$ has a \emph{Naimark dilation} \cite{CBMOA03}, i.e., a triplet $(\hik, \hat{\A}, V)$ where  $\hik$ is a Hilbert space, $V: \hin \to \hik$ is an isometry and $\hat{\A}$ is a sharp observable on $\hik$ satisfying $V^*\hat{\A}(x)V=\A(x)$ for each $x \in \Omega_{\A}$.
Moreover, there exists a \emph{minimal} Naimark dilation, meaning that the set $\{\sum_x c_x \hat{\A}(x) V \psi : c_x \in \complex, \psi \in \hi\}$ is dense in $\hik$.
This is essentially unique in the sense that if $(\hik_1, \hat{\A}_1, V_1)$ is a minimal Naimark dilation and $(\hik_2, \hat{\A}_2, V_2)$ is any other Naimark dilation of $\A$, then there exists an isometry $J:\hik_1\to\hik_2$ satisfying $J\hat{\A}_1(x)=\hat{\A}_2(x)J$ and $JV_1=V_2$. 
In particular, $J$ is a unitary operator if both Naimark dilations are minimal.

\begin{proof}
Fix a minimal Naimark dilation $(\hik, \hat{\A}, V)$ of $\A$.
We define a channel $\Lambda_{\A}$ with input and output Hilbert spaces $\hin$ and $\hik$, respectively, by
\begin{eqnarray}\label{eq:LambdaA} 
\Lambda_{\A} (\varrho) =\sum_x  \hat{\A}(x) V \varrho V^\ast \hat{\A}(x)  \, .
 \end{eqnarray}
This is an $\A$-channel since 
\begin{equation*}
\tr{\hat{\A}(x) V \varrho V^\ast \hat{\A}(x)} = \tr{\varrho V^\ast \hat{\A}(x) V } = \tr{\varrho \A(x)} \, .
\end{equation*}
We claim that $\Lambda_\A$ has the universal property.

First, we recall the preliminary result demonstrated earlier: each $\B$ that is jointly measurable with $\A$ can be written as $\B(y)=\hei{(\Lambda^{\B})}(|y\rangle \langle y|)$, where  $\{|y\rangle \}_{y \in \Omega_{\B}}$ is an orthonormal basis and $\Lambda^{\B}$ is the channel defined in \eqref{eq:A->B}. 
We will show that there exists a channel $\Gamma^\B$ such that 
\begin{equation}\label{eq:circ}
\Lambda^\B=\Gamma^\B \circ \Lambda_\A \, , 
\end{equation}
 where $\circ$ denotes the composition of two functions.
Then, using the Heisenberg form $\hei{(\Gamma^\B)}$ of $\Gamma^\B$, we define an observable $\B'$ as
\begin{equation}\label{eq:b'}
\B'(y)=\hei{(\Gamma^\B)}(|y\rangle \langle y|) \, .
\end{equation}
This observable satisfies the required equation \eqref{eq:universal} since
\begin{align}
\tr{\Lambda_\A(\varrho) \B'(y)}  \stackrel{\eqref{eq:b'}}{=} \tr{\Gamma^\B(\Lambda_\A(\varrho))
|y\rangle \langle y|} \\
 \stackrel{\eqref{eq:circ}}{=} \tr{\Lambda^\B(\varrho) |y\rangle \langle y|}  \stackrel{\eqref{eq:lambdab}}{=} \tr{\varrho \B(y) } \, .
\end{align} 

Hence, to complete the proof we need to show that for each $\B$  that is jointly measurable with $\A$, there exists a channel $\Gamma^\B$ such that \eqref{eq:circ} holds.
Let $\M$ be a joint observable of $\A$ and $\B$, and let $(\hik', \hat{\M}, V')$ be a Naimark dilation of $\M$.
We define a sharp observable $\hat{\A}'$ as $\hat{\A}'(x)=\sum_y \hat{\M}(x,y)$.
Since $\A(x)=\sum_y \M(x,y)$, we observe that $(\hik', \hat{\A}', V')$ is a Naimark dilation of $\A$. 
The initially fixed Naimark dilation $(\hik, \hat{\A}, V)$ was chosen to be minimal, so there exists an isometry $J:\hik\to\hik'$ satisfying $\hat{\A}'(x)J = J\hat{\A}(x)$ and $V'=JV$.
Taking also into account that $\hat{\M}(x,y)\hat{\M}(x',y)=\delta_{xx'}\hat{\M}(x,y)$, we obtain the auxiliary formula
\begin{align}\label{eq:aux}
\hat{\M}(x,y)J\hat{\A}(x')=\delta_{xx'} \hat{\M}(x,y)J \, .
\end{align}
Using \eqref{eq:aux} we can write $\Lambda^\B$ in the form
\begin{align*}
\Lambda^\B(\varrho) = \sum_{x,y} \tr{\hat{\A}(x) V \varrho V^\ast \hat{\A}(x) J^\ast \hat{\M}(x,y) J } \kb{y}{y} \, .
\end{align*}
We define a channel $\Gamma^\B$ as
\begin{align*}
\Gamma^\B(\varrho) = \sum_{x,y} \tr{\varrho J^\ast \hat{\M}(x,y) J} \kb{y}{y} \, .
\end{align*}
Using  \eqref{eq:aux} again one can confirm that the required equation $\Lambda^\B=\Gamma^\B \circ \Lambda_\A$ holds.
\end{proof}

We note that the channel $\Lambda_\A$   was introduced as the least disturbing $\A$-channel in \cite{HeMi13}.
In the context of sequential measurements, the channel $\Lambda_\A$ 
is the appropriate generalization of the L\"uders channel of sharp observables; $\Lambda_\A$ has the universal property and it reduces to the L\"uders channel whenever $\A$ is a sharp observable.
The price we have to pay is the larger output space $\hout$ compared to the input space $\hin$.

\paragraph{\bf{Channels without the universal property.}}

We now turn to a different variation of the problem. 
Suppose two jointly measurable observables $\A$ and $\B$ are given.
We may have limited resources, so perhaps we cannot realize any $\A$-channel having the universal property.
Hence, suppose we first perform a measurement described by a pair $(\A,\Lambda)$, where $\Lambda$ is an $\A$-channel \emph{not} having the universal property.
We want to know if we can still implement a measurement of $\B$, i.e., whether there exists an observable $\B'$ such that
\begin{align}\label{eq:universal-2}
\tr{\Lambda(\varrho) \B'(y)}=\tr{\varrho \B(y) }
\end{align}
for all input states $\varrho$ and outcomes $y\in\Omega_\B$.

To give a general answer to this question, we recall that by the \emph{Stinespring dilation theorem} \cite{CBMOA03} any channel $\Lambda$ can be written in the form
\begin{equation}\label{eq:stinespring-s}
\Lambda(\varrho)=\ptr{\hik}{V\varrho V^*} \, , 
\end{equation}
where $\hik$ is a Hilbert space attached to an environment system, $V:\hin\to\hout\otimes \hik$ is an isometry and $\ptr{\hik}{\cdot}$ is the partial trace over $\hik$.
Tracing over $\hout$ rather than $\hik$ we obtain the corresponding \emph{conjugate channel} (also called complementary channel):
\begin{equation}
\bar{\Lambda}(\varrho)=\ptr{\hout}{V\varrho V^*} \, .
\end{equation}
Hence, the conjugate channel $\bar{\Lambda}$ describes what happens to the environment system during the measurement process described by $\Lambda$.

\begin{theorem}\label{thm:conjugate}
For a channel $\Lambda$ and observable $\B$, there exists an observable $\B'$ satisfying \eqref{eq:universal-2} 
if and only if the conjugate channel $\bar{\Lambda}$ of $\Lambda$ is a $\B$-channel. 
\end{theorem}

\begin{proof}
The Stinespring dilation \eqref{eq:stinespring-s}, written in the Heisenberg picture, reads
\begin{equation}\label{eq:stinespring-h}
\hei{\Lambda}(T) =V^*(T\otimes \id)V \, ,
\end{equation}
and the corresponding conjugate channel is then
\begin{equation}
\hei{\bar{\Lambda}}(S)=V^*(\id \otimes S)V \, .
\end{equation}
As explained in \cite{HeMiRe14}, it follows from the Radon-Nikodym theorem for quantum operations  \cite{Arveson69,Raginsky03} that $\hei{\bar{\Lambda}}$ is a $\B$-channel if and only if there exists an observable $\B'$ such that $V^*(\B'(y) \otimes \id)V =\B(y)$.
Inserting this into \eqref{eq:stinespring-h} gives $\hei{\Lambda}((\B'(y))=\B(y)$, which is \eqref{eq:universal-2}  but in the Heisenberg picture.
\end{proof}

The formulation of Theorem \ref{thm:conjugate} is slightly loose since the existence of $\B'$ may seem to depend on the choice of the conjugate channel. 
However, all conjugate channels of a given channel $\Lambda$ are equivalent, and for this reason the formulation of Theorem \ref{thm:conjugate} is solid.
More specifically, a channel $\Lambda$ has many Stinespring dilations, each of them giving rise to a conjugate channel.
These conjugate channels can be different, but they are equivalent in the sense that each of them can be obtained from any other by concatenating with some other channel. 
This implies that if one conjugate channel of $\Lambda$ is a $\B$-channel, then all conjugate channels of $\Lambda$ are $\B$-channels.
Therefore, it does not matter which conjugate channel we use in Theorem \ref{thm:conjugate}.

\paragraph{\bf{Qubit example.}}

Let us fix $\hin=\complex^2$ and consider two families of binary qubit observables $\A_s$ and $\B_{t,\theta}$, where
\begin{align*}
\A_s(\pm 1) & = \half \bigl( \id \pm s \sigma_z \bigr) \\
\B_{t,\theta}(\pm 1) &= \half \bigl( \id \pm t (\sin\theta\sigma_x +\cos\theta \sigma_z) \bigr) \, ,
\end{align*}
and the parameters belong to the intervals $s,t\in(0,1]$ and $\theta\in[0,\pi/2]$ respectively.
As proved in \cite{Busch86}, $\A_s$ and $\B_{t,\theta}$ are jointly measurable if and only if 
\begin{equation}\label{eq:busch}
s^2 + t^2 - \cos^2\theta s^2 t^2 \leq 1 \, .
\end{equation}
First, we want to see how the pairs satisfying this inequality can be implemented sequentially.

There exists an $\A_s$-channel satisfying the nondisturbance condition \eqref{eq:nondisturbance} for $\B_{t,\theta}$ if and only if $\A_s$ and  $\B_{t,\theta}$  commute \cite[Prop. 6]{HeWo10}, which is the case when $\theta\in\{0,\pi/2\}$.
Therefore, most of the realizable joint measurements must be performed by first measuring $\A_s$, followed by some modified version $\B'_{t,\theta}$ of $\B_{t,\theta}$. 
The modified observable $\B'_{t,\theta}$ and a suitable $\A_s$-channel are not difficult to find in this simple case; we can choose $\Lambda$ to be the L\"uders channel of $\A_s$ and $\B'_{t,\theta}=\B_{1,\theta}$. 
With these choices \eqref{eq:universal-2} is satisfied, hence 
the L\"uders channel can be used to measure all binary observables jointly measurable with $\A$.

Let us then consider another pair of observables in order to  demonstrate that the L\"uders channel of $\A_s$ does not have the universal property.
Fix $0<s<1$ and let $\C$ be the following four outcome observable:
\begin{align*}
& \C(1,\pm1)=\frac{1 \pm s}{4} (1 \pm \sigma_z) \, , \C(-1,1)=\frac{1 - s}{4} (1 + \sigma_x) \\
& \C(-1,-1)=\frac{1 - s}{4} (1 - \sigma_x) + \frac{s}{2} (1-\sigma_z)
\end{align*}
We have $\C(1,1)+\C(1,-1)=\A_s(1)$ and $\C(-1,1)+\C(-1,-1)=\A_s(-1)$, hence $\C$ and $\A_s$ are jointly measurable.
One can utilize Theorem \ref{thm:conjugate} to see that there is no $\C'$ such that 
\begin{equation}\label{eq:cannot-hold}
\luders{\A}^\ast(\C'(j,k))=\C(j,k) \quad \textrm{for all $j,k=\pm 1$}.
\end{equation}
We can also confirm this fact directly; let us make a counter assumption that there exists an observable $\C'$ satisfying \eqref{eq:cannot-hold}.

Since the operator $\C(-1,1)$ is rank-1 and the operators $\A_s(\pm 1)$ are invertible, it follows that $\C'(-1,1)$ is rank-1 as well.
As $\Lambda$ is trace preserving, we conclude that $\Lambda(P)=\half (1 + \sigma_x)$
must hold for some 1-dimensional projection $P$.
A direct calculation shows that $\Lambda(P)$ is a projection if and only if $P=\half (1 \pm \sigma_z)$, in which case we have $\Lambda(P)=P$.
Therefore, $\Lambda(P)=\half (1 + \sigma_x)$ cannot be satisfied by any projection $P$.
As a conclusion, the L\"uders channel of $\A_s$ does not have the universal property.

One can compare the L\"uders channel of $\A_s$ with the universal channel $\Lambda_{\A_s}$ defined in \eqref{eq:LambdaA}. 
The first has $\hout=\complex^2$, while the latter has $\hout=\complex^4$.
The larger output space makes the subsequent measurement able to take advantage also of the `information leaked to the environment', thereby leading to the universality property.

\paragraph{Acknowledgments}
The authors are grateful to David Reeb, Daniel Reitzner and Leon Loveridge for their comments on an earlier version of this paper.
TH acknowledges the financial support from the Academy of Finland (grant no. 138135).
TM thanks JSPS for the financial support (JSPS KAKENHI Grant Numbers 22740078).

\end{document}